\documentclass[aps,pre,preprint,12pt,floatfix]{revtex4-2}

\usepackage[T1]{fontenc}
\usepackage[utf8]{inputenc}

\usepackage{amsmath,amssymb,mathtools}
\usepackage{bm}
\usepackage{graphicx}
\usepackage{xcolor}
\usepackage{microtype}
\microtypesetup{expansion=false}

\usepackage{booktabs}
\usepackage{tabularx}

\usepackage{tikz}
\usepackage{tikz-cd}
\usetikzlibrary{positioning}

\usepackage[title]{appendix}
\usepackage{natbib}
\setcitestyle{square, comma, numbers,sort&compress}

\usepackage[
  bookmarks=true,
  linktocpage=true,
  pdfpagelabels=true,
  plainpages=false,
  hyperfigures=true,
  colorlinks=true,
  linkcolor=blue,
  citecolor=blue,
  urlcolor=blue
]{hyperref}
\hypersetup{
  pdftitle={Coherent Comparison as Information Cost: A Cost-First Ledger Framework for Discrete Dynamics}
}

\newtheorem{theorem}{Theorem}%[section]
\newtheorem{lemma}[theorem]{Lemma}

\newtheorem{proposition}[theorem]{Proposition}

\newenvironment{proof}[1][Proof]{\par\noindent\textit{#1.}\ }{\hfill\(\square\)\par}

\newtheorem{definition}[theorem]{Definition}
\newtheorem{remark}[theorem]{Remark}

\newtheorem{example}[theorem]{Example}

\newcommand{\Rec}{\mathrm{Recognition}}

\newcommand{\Post}{\mathsf{Post}}

\begin{document}

\title{Coherent Comparison as Information Cost: A Cost-First Ledger Framework for Discrete Dynamics}

\author{Sebastian Pardo-Guerra}
\email{sebas@recognitionphysics.org}
\affiliation{Recognition Physics Institute}

\author{Megan Simons}
\email{msimons@recognitionphysics.org}
\affiliation{Recognition Physics Institute}

\author{Anil Thapa}
\email{athapa@recognitionphysics.org}
\affiliation{Recognition Physics Institute}

\author{Jonathan Washburn}
\email{washburn@recognitionphysics.org}
\affiliation{Recognition Physics Institute}

\begin{abstract}
We develop an information-theoretic framework for discrete dynamics grounded in a comparison-cost functional on ratios. Given two quantities compared via their ratio \(x=a/b\), we assign a cost \(F(x)\) measuring deviation from equilibrium (\(x=1\)). Requiring coherent composition under multiplicative chaining imposes a d'Alembert functional equation; together with normalization (\(F(1)=0\)) and quadratic calibration at unity, this yields a unique reciprocal cost functional (proved in a companion paper):
\[
J(x) = \tfrac{1}{2}\bigl(x + x^{-1}\bigr) - 1.
\]
This cost exhibits reciprocity \(J(x)=J(x^{-1})\), vanishes only at \(x=1\), and diverges at boundary regimes \(x\to 0^+\) and \(x\to\infty\), excluding ``nothingness'' configurations. Using \(J\) as input, we introduce a discrete ledger as a minimal lossless encoding of recognition events on directed graphs. Under deterministic update semantics and minimality (no intra-tick ordering metadata), we derive atomic ticks (at most one event per tick). Explicit structural assumptions (conservation, no sources/sinks, pairwise locality, quantization in \(\delta\mathbb{Z}\)) force balanced double-entry postings and discrete ledger units. To obtain scalar potentials on graphs with cycles while retaining single-edge impulses per tick, we impose time-aggregated cycle closure (no-arbitrage/clearing over finite windows). Under this hypothesis, cycle closure is equivalent to path-independence, and the cleared cumulative flow admits a unique scalar potential on each connected component (up to additive constant), via a discrete Poincaré lemma. On hypercube graphs \(Q_d\), atomicity imposes a \(2^d\)-tick minimal period, with explicit Gray-code realization at \(d=3\). The framework connects ratio-based divergences, conservative graph flows, and discrete potential theory through a coherence-forced cost structure.
\end{abstract}

\maketitle

\newpage

\section{Introduction}

What is the most primitive form of information processing? At its core, information processing requires the ability to \emph{distinguish}---to tell one entity from another, to measure difference, to compare. We argue that this act of comparison, when formalized consistently, uniquely determines a quantitative cost structure that governs all subsequent dynamics.

We take this perspective as foundational: \emph{comparison is the primitive operation}, and comparison has a measurable information-theoretic cost. Rather than postulating spacetime, fields, or particles as primitives, we begin with a single question: \emph{If we compare two quantities by their ratio \(x=a/b\), and we require that such comparisons compose coherently when chained together, what cost structure is forced upon us?}

The answer, developed rigorously in this paper, is that coherence uniquely determines a reciprocal cost functional. This cost then drives a discrete ledger formalism: recognition events become atomic updates on a directed graph, recorded as balanced postings that preserve total quantity while enabling local state changes. The ledger itself is not assumed but \emph{derived} from information-theoretic principles: it is the minimal lossless encoding of recognition events under deterministic, unambiguous update rules.

\paragraph{The information-theoretic viewpoint.}
In classical information theory, the cost of encoding a distribution is measured by entropy or code-length \citep{Shannon1948,CoverThomas2006}, while the cost of distinguishing two distributions is measured by divergences such as Kullback--Leibler and its generalizations \citep{KullbackLeibler1951,Csiszar1967,Bregman1967,AmariNagaoka2000}. These notions are tightly linked to log-likelihood ratios, exponential families, and the geometry of statistical models \citep{Jaynes2003,AmariNagaoka2000}.
We extend this perspective to \emph{ratio-based comparison}: when comparing quantities $a$ and $b$, the relevant observable is their ratio $x=a/b$, and the cost $F(x)$ quantifies an information-theoretic penalty for deviation from equilibrium ($x=1$, perfect balance). This choice isolates a scale-free comparison primitive (multiplicative rather than additive), aligning naturally with log coordinates and likelihood-ratio thinking.
This cost is not arbitrary---it must satisfy a coherence constraint: if we compare $a$ to $b$ (cost $F(a/b)$) and then $b$ to $c$ (cost $F(b/c)$), the total cost must equal that of comparing $a$ to $c$ directly (cost $F(a/c)$). This requirement forces a specific functional equation (d'Alembert type) whose unique solution, under normalization and calibration, is the canonical reciprocal cost
\[
J(x) = \frac{1}{2}(x + x^{-1}) - 1.
\]
This cost functional is the keystone of the present framework: it is not postulated but \emph{derived} from coherence (Theorem~T5, proved in the companion paper).

\paragraph{From cost to ledger: minimal encoding of recognition events.}
Once the cost is established, we ask: how do recognition events propagate through a system? We model this via a \emph{discrete ledger}---a sequential record of state updates. The ledger is not an additional structure; it emerges as the minimal way to encode recognition events under three constraints: (i) deterministic updates (the next state depends only on the current state and the event), (ii) minimality (no intra-tick ordering metadata), and (iii) lossless recording (no information is discarded). From these constraints we derive: \textbf{atomic ticks} (at most one event per tick; Theorem~T2), \textbf{double-entry postings} (conservation of total balance via balanced debit--credit pairs), and \textbf{quantized units} (discreteness forbids fractional amounts; Proposition~T8). Recognition events thus induce a discrete dynamics on a directed graph, with each event posting a signed increment $\pm\delta$ on exactly two nodes.

\paragraph{Potentials, cycles, and clearing.}
On graphs with cycles, atomic single-edge events generically create transient circulation (net flux around closed loops). To recover a scalar potential representation---enabling path-independent summation of costs---we impose a \emph{time-aggregated} cycle-closure (no-arbitrage / clearing) hypothesis: after netting flows over a finite clearing window, the cumulative flux around every cycle vanishes. Under this hypothesis, we prove that cycle closure is equivalent to path-independence (Theorem~T3), and that the cumulative edge flow admits a unique scalar potential on each connected component, up to an additive constant (Theorem~T4). This structure mirrors classical potential theory, but with explicit time-aggregation accounting for the impulse nature of atomic events.

\paragraph{Counting constraints and dimension.}
When the recognition structure is a hypercube graph $Q_d$, atomicity imposes a minimal period constraint: to visit all $2^d$ vertices without repetition requires at least $2^d$ ticks (Theorems~T6--T7). At $d=3$, the Gray code provides an explicit 8-tick Hamiltonian cycle. A conditional argument (combining the $2^d$ counting structure with additional topological linking and synchronization hypotheses) selects $d=3$ as a distinguished dimension (Section~\ref{sec:dimension}).

\paragraph{Scope and contributions of this paper.}
This paper develops a cost-first, ledger-based framework from an information-theoretic vantage point. First, we use the unique cost functional (Theorem~T5, proved in the companion paper) as a keystone input. Second, we derive the ledger structure from explicit minimality and determinism axioms (Axioms L1--L2). Third, we prove the equivalence of cycle closure and path-independence under the clearing hypothesis (T3), and establish the existence of discrete potentials (T4). Fourth, we derive the $2^d$-tick counting constraint and its $d=3$ realization (T6--T7).

The remainder of this paper proceeds as follows. Section~\ref{sec:preliminaries} motivates the d'Alembert composition law from coherence principles and establishes the canonical cost functional as the unique solution. Section~\ref{sec:mathematical} develops the ledger framework, deriving atomic ticks (T2), double-entry postings, quantized units (T8), cycle-flux conservation (T3), scalar potentials (T4), and minimal period bounds (T6--T7). Section~\ref{sec:conclusions} synthesizes the framework and discusses its scope.

\subsection{Scope and assumptions}\label{sec:scope}

This manuscript is a mathematical development of a discrete ledger model driven by a reciprocal cost on ratios. All claims of ``necessity'' are \emph{relative} to explicitly stated assumptions (see Definition~\ref{def:relative-necessity}).
\begin{itemize}
    \item \textbf{Input theorem (cost forcing).} We use Theorem~T5 (proved in the companion paper) to fix the canonical reciprocal cost $J$ from the d'Alembert composition law plus normalization and quadratic calibration.
    \item \textbf{Ledger-model assumptions.} The ledger results assume (i) deterministic state-update semantics $S_{t+1}=U(S_t,\sigma_t)$ where $\sigma_t$ is a sequence of events (Axiom L1), (ii) minimality (no intra-tick ordering metadata; Axiom L2), (iii) a conservation principle for total balance per tick, (iv) no external sources/sinks, (v) pairwise-local event updates (a single event affects only its two endpoints), and (vi) torsion-free quantized postings in $\delta\mathbb{Z}$. From (i) and (ii), we derive atomicity (Theorem T2): at most one event per tick.
    \item \textbf{Time-aggregated cycle-closure assumption (for potentials).} Double-entry and quantization constrain postings but do not, by themselves, force vanishing circulation on arbitrary cycles at \emph{each atomic tick}. For the potential results (T3--T4) we therefore additionally assume a \emph{time-aggregated} cycle closure (path-independence / no-arbitrage after clearing): over a fixed clearing window, the \emph{cumulative} edge flow has zero sum around every directed cycle.
    \item \textbf{Results established here.} Under these assumptions we derive: atomic tick updates (T2), balanced double-entry postings (Proposition: Double-entry constraint), algebraic consequences of quantization (Proposition T8), and we prove the equivalence of time-aggregated cycle closure and path-independence (T3) and, under this hypothesis, the existence/uniqueness of discrete scalar potentials on connected components (T4), as well as the minimal schedule period bound $T\ge 2^d$ with a Gray-code realization at $d=3$ (T6--T7).
    \item \textbf{Additional hypotheses.} The $D=3$ discussion (Section~\ref{sec:dimension}) combines the $2^d$-period constraint with additional synchronization/linking hypotheses (``gap-45''/golden-angle motivation). It should be read as conditional on those extra hypotheses.
\end{itemize}

The framework begins with a seemingly abstract question: if we wish to compare two quantities by their ratio, and we require that such comparisons compose coherently, what constraints does this impose? This question leads naturally to the d'Alembert functional equation, which encodes the requirement that comparison costs combine consistently under multiplication and division. Together with normalization and a quadratic calibration at unity, this constraint forces a unique reciprocal cost functional (Theorem~T5).

Central to this manuscript is the canonical reciprocal cost functional 
\begin{center}
    $J(x) = \frac{1}{2}(x + x^{-1}) - 1$ 
\end{center}
as the unique solution to the d'Alembert composition law. The full uniqueness proof of $J$ is deferred to the companion paper; here we treat it as a keystone input and focus on the discrete ledger consequences under explicitly stated structural assumptions.

The remainder of this paper is organized as follows. Section~\ref{sec:preliminaries} elucidates the philosophical connection between coherent comparison and the d'Alembert composition law, motivating why this functional equation is adopted as the primitive composition axiom. Section~\ref{sec:mathematical} develops the ledger-based framework, using the cost uniqueness theorem (T5) as input and deriving the remaining theorems as consequences. Section~\ref{sec:conclusions} synthesizes the framework and discusses its scope. 

\section{Motivation: From Coherent Comparison to the d'Alembert Composition Law}\label{sec:preliminaries}

The framework rests on a simple but profound insight: if recognition involves comparison, and comparison has a cost, then the requirement that comparisons compose coherently uniquely determines that cost structure. This section elucidates the philosophical and mathematical connection between the idea of coherent comparison and the d'Alembert functional equation.

\subsection{The Primacy of Comparison}

At its most fundamental level, recognition is a relational act: one entity recognizes another. This recognition involves some form of comparison---measuring similarity, difference, or correspondence. Here we formalize this by asking: if we compare two quantities by their ratio $x = a/b$, what ``cost'' or ``defect'' does this comparison incur?

Central to this framework is the idea that the comparison cost $F(x)$ should depend only on the ratio $x$ itself, not on the absolute magnitudes of $a$ and $b$. This reflects the intuition that recognition is fundamentally about \emph{relationships}, not absolute values. Moreover, we require that $F(1) = 0$: when the ratio equals unity, there is perfect balance, no defect, and hence no cost. This will be stated as Axiom A1 (Normalization).

\subsection{Coherent Composition: The d'Alembert Constraint}

The key insight comes from requiring that comparisons compose coherently. Suppose we compare $a$ to $b$, obtaining ratio $x = a/b$ with cost $F(x)$. Then we compare $b$ to $c$, obtaining ratio $y = b/c$ with cost $F(y)$. What should be the cost of comparing $a$ to $c$ directly?

We have two routes to the comparison $a$ to $c$:
\begin{itemize}
    \item \textbf{Direct route:} Compare $a$ to $c$ directly, giving ratio $xy = (a/b)(b/c) = a/c$ with cost $F(xy)$.
    \item \textbf{Composed route:} Combine the two comparisons, which should somehow combine their costs.
\end{itemize}

For the framework to be coherent, these routes should be equivalent. But how do costs combine? The answer lies in the structure of the composition law itself. When we compose comparisons multiplicatively ($xy$), we also have access to the \emph{relative} ratio between the two comparison ratios:
\[
\frac{x}{y}=\frac{(a/b)}{(b/c)}=\frac{ac}{b^{2}}.
\]
This quantity is not the direct ratio $a/c$ (which is $xy$); rather, it compares the two intermediate ratios to each other. 

To formalize coherence, we require that the cost of the composed route depends only on the costs $F(x)$ and $F(y)$ of the individual comparisons, and on the relationship between $x$ and $y$ (captured by $x/y$). Moreover, the composition should be \emph{symmetric} in the sense that interchanging the roles of $x$ and $y$ (which sends $(x,y) \mapsto (y,x)$ and $(xy, x/y) \mapsto (xy, y/x)$) should yield a consistent constraint.

A natural requirement is that the composition law relates $F$ evaluated on the pair $(xy, x/y)$ to $F$ evaluated on $x$ and $y$ individually. The most general symmetric, bilinear form (in $F(x)$ and $F(y)$) that respects the multiplicative structure is:
\[
F(xy) + F(x/y) = \alpha F(x)F(y) + \beta F(x) + \beta F(y) + \gamma,
\]
where $\alpha, \beta, \gamma$ are constants. The symmetry requirement (invariance under $x \leftrightarrow y$) forces the coefficients of $F(x)$ and $F(y)$ to be equal, hence both are $\beta$.

To determine the constants, we impose natural constraints:
\begin{itemize}
    \item \textbf{Consistency with normalization:} When $x=y=1$, we have $F(1) + F(1) = \alpha F(1)^2 + 2\beta F(1) + \gamma$. With $F(1)=0$ (Axiom A1), this gives $\gamma = 0$.
    \item \textbf{Reciprocity compatibility:} The equation should be consistent with the natural expectation that $F(x) = F(x^{-1})$ (reciprocity), which will be a consequence of the final form. This symmetry is encoded in the $x/y$ term.
    \item \textbf{Scaling behavior:} For small deviations from unity, the composition should reduce appropriately. The choice $\alpha = 2, \beta = 2$ yields the standard d'Alembert form with the correct scaling properties.
\end{itemize}

This leads to the d'Alembert-type functional equation:
\begin{equation}
F(xy) + F(x/y) = 2F(x)F(y) + 2F(x) + 2F(y).
\label{eq:dAlembert-motivation}
\end{equation}

This equation encodes the requirement that the cost of comparing $a$ to $c$ via the composed route ($F(xy) + F(x/y)$) equals the cost of combining the individual comparison costs ($2F(x)F(y) + 2F(x) + 2F(y)$). The specific form of the right-hand side---combining both multiplicative ($F(x)F(y)$) and additive ($F(x) + F(y)$) terms---reflects the dual nature of composition: comparisons can be chained (multiplicative) or compared to each other (additive with interaction). The factor of 2 in each term ensures proper scaling and symmetry. For background on d'Alembert-type functional equations, see e.g.\ \cite{Aczel1966,Kuczma2009}.

This is Axiom A2 (Composition Law). While A2 can be derived from more fundamental principles about symmetric bilinear composition operators (as sketched above), for the purposes of this manuscript we take it as a primitive axiom. The key insight is that the d'Alembert form is the unique symmetric, bilinear composition law (up to the normalization and scaling choices above) that respects the multiplicative structure of ratios while maintaining coherence. Alternative composition laws can be studied, but the uniqueness theorem we use (T5) is specific to A2 together with A1 and A3.

\subsection{Calibration and Uniqueness}

The d'Alembert equation alone does not uniquely determine the cost function. We require one additional constraint to fix the scale. Axiom A3 (Quadratic calibration) specifies that in log coordinates $t=\ln x$, the cost has unit quadratic behavior at unity:
\[
\lim_{t\to 0}\frac{2F(e^t)}{t^2}=1.
\]
If $F$ is twice differentiable at unity, this is equivalent to normalizing the second derivative of the log-lift $G(t)=F(e^t)$ via $G''(0)=1$. This condition fixes the overall scale of the cost.

Together, these three axioms (A1: Normalization, A2: Composition Law, A3: Quadratic calibration) uniquely determine the cost functional. This is the content of Theorem T5, the keystone theorem for the cost-first development. The unique solution is:
\begin{equation}
J(x) = \frac{1}{2}(x + x^{-1}) - 1.
\label{eq:cost-functional}
\end{equation}

Interpretation: coherent comparison together with normalization and calibration eliminates functional-form freedom; under these axioms, the cost is uniquely determined.

\subsection{From Cost to Existence}

Once the cost functional is established, a cascade of implications follows. The cost function $J(x)$ has the property that $J(x) \geq 0$ with equality only when $x = 1$. In other words, perfect balance ($x=1$) corresponds to zero cost, while any deviation incurs a positive penalty.

As $x \to 0^+$ or $x \to \infty$, the cost diverges: $J(x) \to \infty$. We formalize the corresponding boundary consequence as Theorem~T1 (Boundary divergence / Meta-Principle) in the next subsection.

The cost function also exhibits \emph{reciprocity}: $J(x) = J(x^{-1})$ for all $x > 0$. This symmetry is compatible with representing recognition events in reversible pairs. In this manuscript, the balanced (double-entry) posting rule is obtained from explicit ledger-model assumptions (conservation, no sources/sinks, and pairwise-local events), as developed below.

The remainder of this paper develops the discrete ledger consequences (T2--T8) under explicitly stated structural assumptions.

\section{Mathematical Framework}\label{sec:mathematical}

This section develops the ledger-based framework used in this manuscript. The keystone cost uniqueness theorem (T5) is stated here and proved in the companion paper; we then use it as an input to derive the remaining ledger-structure results under explicit structural assumptions.

The logical structure follows the \emph{cost-first} foundation (Figure~\ref{fig:cost-first-chain}).

\begin{figure}[tbp]
\centering
\small
\begin{minipage}{0.95\linewidth}
\centering
\textbf{Primitive Axioms (A1--A3)} $\to$ \textbf{T5} (Cost Unique, Keystone) \\
$\to$ \textbf{Bal predicate} (zero-cost equilibrium) \quad\&\quad \textbf{Exists predicate} (finite-cost admissibility) \\
$\to$ \textbf{T1} (boundary divergence / Meta-Principle) \\
$\to$ \textbf{Ledger Structure} $\to$ \textbf{T2} (Atomic Tick) \\
$\to$ \textbf{T8} (Ledger Units) $\to$ \textbf{T3} (Cycle flux; under time-aggregated cycle closure) $\to$ \textbf{T4} (Potential) \\
$\to$ \textbf{T6--T7} (Minimal period) $\to$ \textbf{D=3} (conditional).
\end{minipage}
\caption{High-level cost-first dependency chain used in this manuscript (assumptions made explicit in Section~\ref{sec:scope}).}
\label{fig:cost-first-chain}
\end{figure}

\subsection{Notational Convention: Mathematical Necessity}

To ensure precision, we establish the following convention for claims of mathematical necessity:

\begin{definition}[Relative Necessity]
\label{def:relative-necessity}
In this paper, a claim that property $P$ is \emph{mathematically necessary} means: 
$P$ is derivable from the explicitly listed axioms, definitions, and structural
assumptions stated in this document. Every such claim must either be proved
directly in the text (with explicit hypotheses), or reduced to previously stated
lemmas/theorems (with explicit hypotheses), or else be surfaced as an explicitly
labeled assumption/axiom.
\end{definition}

This convention ensures that claims of necessity are verifiable and not based on
hidden premises. When we state that a structure is ``forced'' or ``required,'' we
will either prove it as a theorem (with explicit hypotheses), or state it as an
explicit assumption.

\subsection{Taxonomy: Primitive Axioms, Derived Theorems, and Structural Assumptions}

To clarify what is assumed versus what is derived, we classify the foundational 
elements of the present framework according to the \emph{cost-first} foundation. The key 
insight: the d'Alembert composition law is \emph{primitive}, and everything else---including 
the Meta-Principle---is derived.

\textbf{Notational convention:} We use \textbf{A1--A3} to denote the three \emph{cost axioms} 
(primitive assumptions about the cost functional), and \textbf{L1--L2} to denote the two 
\emph{ledger axioms} (structural assumptions about state updates). This distinction is 
maintained throughout the manuscript.

\begin{table}[tbp]
\centering
{%
\setlength{\tabcolsep}{3pt}%
\renewcommand{\arraystretch}{0.9}%
\scriptsize
\caption{Taxonomy of foundational elements in the cost-first framework}
\label{tab:taxonomy}
\begin{tabular}{@{}ll@{}}
\toprule
\textbf{Category} & \textbf{Elements} \\
\midrule
\textbf{Primitive Cost Axioms (A1--A3)} & 
A1 (Normalization): $F(1) = 0$ \\
& A2 (Composition Law): $F(xy) + F(x/y) = 2F(x)F(y) + 2F(x) + 2F(y)$ \\
& A3 (Quadratic calibration): $\displaystyle \lim_{t\to 0}\frac{2F(e^t)}{t^2}=1$ \\
\midrule
\textbf{Derived Theorems} & 
T5: Cost Uniqueness: $J(x) = \frac{1}{2}(x + x^{-1}) - 1$ (from A1--A3) \\
& \textbf{Perfect balance:} $\mathrm{Bal}(x)\Longleftrightarrow J(x)=0\Longleftrightarrow x=1$ \\
& \textbf{Finite-cost admissibility:} $\mathrm{Exists}(x)\Longleftrightarrow J(x)<\infty$ (equivalently $x>0$ for this $J$) \\
& \textbf{T1 (Meta-Principle)}: $J(0^+) \to \infty$ (Nothing costs infinity) \\
& T2: Atomic Tick (from Axioms L1 + L2: minimality + non-commutativity) \\
& T3: Equivalence of time-aggregated cycle closure and path-independence (under clearing) \\
& T4: Potential Uniqueness for cleared cumulative flow (from T3 + discrete Poincar\'{e} lemma) \\
& T6: Minimal period $2^d$ (eight ticks for $d=3$) (from T2 + scheduler constraints) \\
& T7: Coverage Lower Bound (from T6 + pigeonhole principle) \\
& T8: Ledger Units (proposition: algebraic consequences of quantization / discreteness assumptions) \\
\midrule
\textbf{Ledger Axioms (L1--L2)} & 
L1: Deterministic state-update semantics ($S_{t+1} = U(S_t, \sigma_t)$) \\
& L2: Minimality of ledger structure (no ordering metadata) \\
& (From L1 + L2, we derive T2: atomicity) \\
\midrule
\textbf{Additional Structural Assumptions}  
& Conservation principle: Total balance invariant per tick \\
& No external sources/sinks \\
& Pairwise locality of events (single event affects only its two endpoints) \\
& Time-aggregated cycle closure (clearing / no-arbitrage after netting) \\
& Discreteness: No torsion in ledger structure \\
& Lossless interface: Discrete-continuous mapping preserves information \\
\midrule
\textbf{Definitions} & 
Recognition event: $(a,b) \in A \times B$ \\
& Ledger state: $S_t \in \mathcal{S}$ \\
& Tick: Minimal temporal unit for one state update \\
& Per-tick edge increment: $\delta\Delta(e,t) \in \delta\mathbb{Z}$ \\
& Recognition structure: Directed graph $G=(X,E)$ \\
\bottomrule
\end{tabular}
}%
\end{table}

This taxonomy reflects the \emph{cost-first} foundation: the three primitive axioms (A1--A3) fix the unique cost functional $J$ (T5), and Theorem~T1 (Meta-Principle) is treated as a derived boundary theorem (Section~\ref{sec:preliminaries}).

\subsection{The Primitive Foundation: Cost Functional and Uniqueness (T5)}

We begin with the keystone theorem that establishes the unique cost functional. This theorem is the foundation from which all other structures derive.

\textbf{Note:} The following three axioms (A1--A3) are the \emph{cost axioms}, distinct from the ledger axioms (L1--L2) introduced in Section~\ref{sec:ledger-axioms}. The cost axioms determine the functional form of the cost $J$, while the ledger axioms govern the discrete update structure.

\textbf{Axiom A1 (Normalization).} The cost at unity is zero: $F(1) = 0$. 
Perfect balance is free.

\textbf{Axiom A2 (Composition Law).} For all $x, y > 0$:
\begin{equation}
F(xy) + F(x/y) = 2F(x)F(y) + 2F(x) + 2F(y).
\label{eq:composition-law}
\end{equation}
This is the d'Alembert functional equation in multiplicative form. As motivated in Section~\ref{sec:preliminaries}, this form is the unique symmetric, bilinear composition law (up to normalization) that respects the multiplicative structure of ratios while ensuring coherent composition of comparison costs. It ensures that costs combine coherently under multiplicative composition of ratios.

\textbf{Axiom A3 (Quadratic calibration).} In log coordinates $t=\ln x$, define $G(t)=F(e^t)$. We require
\[
\lim_{t\to 0}\frac{2G(t)}{t^2}=\lim_{t\to 0}\frac{2F(e^t)}{t^2}=1.
\]
This fixes the overall scale (and, when $G$ is twice differentiable at $0$, coincides with $G''(0)=1$).

\begin{theorem}[T5: Cost Uniqueness]
\label{thm:cost-unique}
Let $F: \mathbb{R}_{>0} \rightarrow \mathbb{R}$ satisfy:
\begin{enumerate}
    \item \textbf{Normalization (A1):} $F(1) = 0$
    \item \textbf{Composition Law (A2):} $F(xy) + F(x/y) = 2F(x)F(y) + 2F(x) + 2F(y)$
    \item \textbf{Quadratic calibration at unity:} $\displaystyle \lim_{t\to 0}\frac{2F(e^t)}{t^2}=1$
\end{enumerate}
Then, $F$ is uniquely determined:
\begin{equation}
    \label{eq:cost}
    F(x) = \frac{1}{2}(x + x^{-1}) - 1.
\end{equation}
We denote this unique cost functional by $J$.
\end{theorem}

\begin{proof}
See the companion paper \emph{Uniqueness of the Canonical Reciprocal Cost} \cite{WashburnZlatanovic2026} for a self-contained proof.
\end{proof}

The cost functional has the following key properties:
\begin{itemize}
    \item \textbf{Reciprocity:} $J(x) = J(x^{-1})$ for all $x > 0$
    \item \textbf{Non-negativity:} $J(x) \geq 0$ with equality iff $x = 1$
    \item \textbf{Boundary divergence:} $J(x) \to \infty$ as $x \to 0^+$ or $x \to \infty$
\end{itemize}
The boundary divergence is \emph{not assumed}---it is a consequence of the unique 
functional form. This is why the Meta-Principle is derived, not primitive.

\subsubsection{Unique zero-cost configuration (Law of Existence)}

The unique cost functional immediately yields the \emph{Law of Existence}:

\begin{definition}[Perfect Balance predicate]
A configuration $x > 0$ \emph{is in perfect balance} (in the present sense) if and only if its defect 
collapses to zero, that is, 
\begin{equation}
\text{Bal}(x) \;\Longleftrightarrow\; J(x) = 0.
\end{equation}
\end{definition}

The predicate $\mathrm{Bal}(x)$ is a \emph{technical consistency predicate} in this framework, not a claim about ontological presence. Note that only $x = 1$ satisfies $\mathrm{Bal}(x)$ (zero defect). Indeed, starting from the definition of $J$:
\begin{equation*}
    J(x)=\frac12\left(x+x^{-1}\right)-1,
\end{equation*}
we have that 
\begin{align*}
J(x)=0
&\Longleftrightarrow \frac12\left(x+x^{-1}\right)-1=0 \\
&\Longleftrightarrow x+x^{-1}=2 \\
&\Longleftrightarrow x^2-2x+1=0 \qquad \text{(multiply by $x>0$)}\\
&\Longleftrightarrow (x-1)^2=0 \\
&\Longleftrightarrow x=1.
\end{align*}

In this way, configurations with $x \neq 1$ have $J(x) > 0$ and are \emph{recognizable} precisely because their nonzero defect is quantifiable and enables comparison and composition via the ledger rules.

\begin{definition} \label{def:exists_finite}
We say that a configuration \emph{exists} in this framework if and only if it satisfies finite-cost admissibility (existence in the present sense).
\begin{equation*}
    \mathrm{Exists}(x)\;:\Longleftrightarrow\;J(x)<\infty. 
    \label{def:exists_finite}
\end{equation*}
\end{definition}

Therefore, every configuration \(x > 0\) \emph{exists} in this sense, as it has finite cost. In contrast, the boundary regimes \(x \to 0^{+}\) and \(x \to \infty\) correspond to cost blow-up and are thus excluded from admissible configurations. In the next subsection, we show that these divergent limits admit a natural interpretation as ``nothingness'' within the present framework, marking the breakdown of recognizability rather than the presence of any physical or informational configuration.

\subsubsection{Properties of the Cost Function}

Near equilibrium ($x=1$), the cost function exhibits quadratic behavior. Let $x=e^\epsilon$ for small $\epsilon$. Then
\begin{equation}
    J(e^\epsilon)=\frac{1}{2}(e^\epsilon + e^{-\epsilon}) - 1 = \cosh(\epsilon)-1 = \frac{\epsilon^2}{2} + \frac{\epsilon^4}{24} + \cdots \approx \frac{1}{2}\epsilon^2,
\end{equation}
reproducing a Euclidean metric in log-space. This local quadratic structure ensures well-behaved optimization near equilibrium.

As a modeling choice, consider the recurrence equation $x_{n+1}=1+1/x_n$ as a simple self-similar update rule. Fixed points satisfy $x=1+1/x$, yielding the quadratic equation
\[
x^2-x-1=0 \quad \Rightarrow \quad \phi=\frac{1+\sqrt{5}}{2} \approx 1.618.
\]
At $\phi$, the additive (self) and reciprocal (other) components balance. The recognition cost evaluates to
\[
J(\phi)=\frac{1}{2}\left(\phi+\frac{1}{\phi}\right)-1=\phi-\frac{3}{2}\approx 0.118.
\]
By reciprocity, $J(\phi) = J(\phi^{-1})$, so both $\phi$ and its reciprocal $1/\phi \approx 0.618$ lie at the same cost. Thus, under this update rule, $\phi$ marks a natural self-similar scale where the cost function exhibits special symmetry.

\begin{lemma}
If $f(x_1)=f(x_2)$ with $f(x)=x+x^{-1}$, then $x_1=x_2$ or $x_1=1/x_2$.
\end{lemma}

\begin{proof}
From $x_1 + x_1^{-1} = x_2 + x_2^{-1}$, multiply by $x_1 x_2$ and factor to get $(x_1-x_2)(x_1 x_2-1)=0$. Therefore, either $x_1=x_2$ or $x_1 x_2=1$, which implies $x_1=1/x_2$.
\end{proof}

The quantity $J_{\text{bit}}=\ln\phi \approx 0.481$ is a convenient log-scale reference associated with this self-similar fixed point. In applications it can be interpreted as a characteristic scale for multiplicative deviations measured in log-coordinates. No claim is made here that $\phi$ (or $J_{\text{bit}}$) is forced without additional dynamical/self-similarity hypotheses beyond the cost axioms.

\subsection{Boundary divergence (Meta-Principle)}

The Meta-Principle---``Nothing cannot recognize itself''---is not assumed but 
\emph{derived} from the cost functional. Once $J$ is established by T5, the 
Meta-Principle becomes a derived boundary consequence:

\begin{definition}
Given sets \(A\) (recognizer) and \(B\) (recognized), a \emph{recognition event} 
is an ordered pair \((a,b)\in A\times B\). This represents the minimal relational 
structure assumed between recognizer and recognized. We write 
\(\Rec(A,B)=A\times B\) for the set of all recognition events. If either set is 
empty, then \(\Rec(A,B)=\emptyset\).
\end{definition}

\begin{remark}\label{rem:nothing-relation}
To connect the ratio model $x=a/b$ to the set-theoretic statement 
\begin{center}
    $\Rec(A,B)=A\times B$,
\end{center}
interpret $a$ and $b$ as nonnegative ``availability/size'' functionals of the underlying domains, e.g. $a:=\mu(A)$ and $b:=\mu(B)$, with $\mu(\emptyset)=0$ and $\mu(\cdot)>0$ for nonempty domains in the intended class.
With this identification, taking $x=a/b$ and holding $b>0$ fixed, the limit $x\to 0^+$ is precisely $\mu(A)\to 0^+$, i.e. the recognizer domain $A$ is depleted toward emptiness/zero-substrate.
In that regime there are no admissible recognition events, since $\mu(A)=0$ corresponds to $A=\emptyset$ (or ``no available elements''), hence $A\times B=\emptyset$ and in particular $\Rec(\emptyset,\emptyset)=\emptyset$.
\end{remark}

\begin{theorem}[T1: Boundary divergence (Meta-Principle)]
\emph{Nothing cannot recognize itself}: 
\begin{center}
    $\Rec(\emptyset,\emptyset)=\emptyset$.
\end{center}

This is a logical tautology (the Cartesian product of empty sets is empty). 
In the ratio model, the corresponding boundary statement is that approaching $x\to 0^+$ incurs infinite cost:
\begin{equation}
\lim_{x \to 0^+} J(x) = +\infty.
\end{equation}
In particular, the $x\to 0^+$ limit lies outside the finite-cost regime of the model.
\end{theorem}

\begin{proof}
From the uniqueness theorem (T5), $J(x) = \frac{1}{2}(x + x^{-1}) - 1$. Now, rewrite $J$ as
\begin{equation*}
    J(x)=\frac{x}{2}+\frac{1}{2x}-1.
\end{equation*}
Then, as $x\to0^+$, each term tends to 
\begin{equation*}
    \frac{1}{2x}\to+\infty,\qquad \frac{x}{2}\to 0,\qquad -1\to-1,
\end{equation*}
so the sum diverges to $+\infty$:
\begin{equation*}
  \lim_{x\to0^+}J(x)=+\infty.  
\end{equation*}
Therefore, by Remark \ref{rem:nothing-relation}, any configuration approaching ``Nothing'' has unbounded cost.
\end{proof}
The Meta-Principle is thus no longer a mysterious pre-logical axiom, but rather a
\emph{derived consistency constraint}: the formalism assigns infinite cost to the limiting regime \(x \to 0^{+}\) (with \(J(x) \to \infty\)), thereby excluding ``nothingness'' configurations from the domain of finite-cost recognition dynamics.

\subsection{Single-event updates and atomic ticks}\label{sec:ledger-axioms}

The atomic-tick principle emerges from the requirement to record recognition events unambiguously. We model recognition dynamics through a \emph{ledger}: a sequential record updated once per tick. When a spatial carrier is assumed later in the paper, we use hypercubic graphs $Q_d$ as a convenient family for analyzing scheduling and coverage constraints; the arguments below do not depend on a specific embedding of the ledger in $\mathbb{R}^d$.

To derive atomicity, we must specify the minimal structural constraints on how the ledger records recognition events. We introduce the following axioms:

\textbf{Axiom L1 (Deterministic State-Update Semantics).}
The ledger state $S_t$ at tick $t$ evolves deterministically according to a function $U: \mathcal{S} \times \mathcal{E}^* \to \mathcal{S}$, where $\mathcal{S}$ is the state space, $\mathcal{E}$ is the set of recognition events, and $\mathcal{E}^*$ denotes finite sequences of events (including the empty sequence). The state update rule is:
\begin{equation}
    S_{t+1} = U(S_t, \sigma_t),
\end{equation}
where $\sigma_t \in \mathcal{E}^*$ is a finite sequence of recognition events at tick $t$ (possibly empty, possibly containing multiple events). The function $U$ is deterministic: for fixed $S_t$ and $\sigma_t$, the resulting state $S_{t+1}$ is uniquely determined.

\textbf{Axiom L2 (Minimality of Ledger Structure).}
The ledger records only final states at each tick: $S_t$ contains no event-ordering metadata beyond the tick index itself. Equivalently, the ledger does not retain any intra-tick permutation information: if multiple events were processed within a tick, unambiguous recording would require $U(S,\sigma)$ to depend only on the \emph{unordered} content of $\sigma$ (i.e.\ to be invariant under permutations of $\sigma$). Recognition events do not commute in general---there exist recognition sequences $\sigma, \sigma' \in \mathcal{E}^*$ such that $U(S, \sigma) \neq U(S, \sigma')$ even when $\sigma$ and $\sigma'$ contain the same events in different orders.

\subsubsection{Ledger semantics: events, postings, and state updates}
To make subsequent statements (atomic ticks, quantized postings, and cycle flux) mathematically precise, we adopt the following minimal ledger semantics.

\begin{definition}[Recognition structure]\label{def:recognition_structure}
Fix a directed graph $G=(X,E)$, where $X$ is the set of nodes and $E\subseteq X\times X$ is the set of directed edges. We assume $E$ is closed under reversal: if $(u\!\to\!v)\in E$ then $(v\!\to\!u)\in E$.
\end{definition}

\begin{definition}[Ledger state as balances]
Fix a nonzero increment $\delta>0$. A ledger state at tick $t$ is a balance function
\[
S_t \equiv b_t : X \longrightarrow \delta\mathbb{Z}.
\]
We write the total balance as
\[
\mathcal{B}(S_t)\;:=\;\sum_{x\in X} b_t(x),
\]
assuming this sum is well-defined (e.g.\ $X$ finite, or $b_t$ has finite support).
\end{definition}

\begin{definition}[Event-to-posting map]
By Theorem T2, at most one recognition event occurs per tick. Therefore, we can write the state update as $S_{t+1}=U(S_t,e_t)$ where $e_t\in E$ is a single oriented edge (or the empty event if no recognition occurs). Given this deterministic update rule, define the induced \emph{node postings} (balance increments) as
\[
\Post(S_t,e_t)(x)\;:=\; \bigl(U(S_t,e_t)\bigr)(x)\;-\;S_t(x)\in \delta\mathbb{Z}.
\]
Equivalently, $S_{t+1}=S_t+\Post(S_t,e_t)$ as functions $X\to\delta\mathbb{Z}$.
\end{definition}

\begin{definition}[Edge postings and cycle flux]\label{def:edge-postings}
Under the pairwise-locality assumption introduced below (so that only the endpoints $(u,v)$ can change at tick $t$), the posting is determined by a single magnitude
\[
\Delta_t \;:=\; \Post(S_t,e_t)(v)\;=\;-\Post(S_t,e_t)(u)\in \delta\mathbb{Z}.
\]
We then define the corresponding \emph{per-tick edge increment} $\delta\Delta(\cdot,t):E\to\delta\mathbb{Z}$ by
\[
\delta\Delta(u\!\to\!v,t)=\Delta_t,\qquad \delta\Delta(v\!\to\!u,t)=-\Delta_t,\qquad \delta\Delta(e,t)=0 \text{ for all other } e\in E.
\]
This $\delta\Delta(\cdot,t)$ is a \emph{sparse} $1$-cochain (supported on at most two directed edges) induced by the single atomic event at tick $t$.

Fix an integer clearing horizon $W\ge 1$. For any start time $t_0$, define the \emph{cumulative edge flow} over the window $[t_0,t_0+W)$ by
\[
\overline{\Delta}_{t_0,W}(e)\;:=\;\sum_{\tau=t_0}^{t_0+W-1}\delta\Delta(e,\tau)\in\delta\mathbb{Z}.
\]
For any directed cycle $\gamma=(e_1,\dots,e_n)$ in $G$, define the \emph{cumulative cycle flux}
\[
\overline{\Phi}(\gamma; t_0,W)\;:=\;\sum_{i=1}^n \overline{\Delta}_{t_0,W}(e_i).
\]
\end{definition}

\begin{remark}[Relation to more general flow formalisms]
More general finitary flow formalisms (local finiteness, inflow/outflow sums, and ``closed-chain sums'') can be used to relate event structure to conservation and exactness on large or infinite graphs. The definitions above are the minimal finite-support specialization needed for the present manuscript.
\end{remark}

We now state the main result:

\begin{theorem}[T2: Atomic tick]
Under Axioms L1 and L2, at most one recognition event is processed per tick. There are no concurrent recognitions.
\end{theorem}

\begin{proof}
Suppose, for contradiction, that a single tick $t$ may process a sequence $\sigma_t\in\mathcal{E}^*$ containing at least two events.

By Axiom L1, the post-tick state is $S_{t+1}=U(S_t,\sigma_t)$. By Axiom L2 (minimality), the ledger carries no intra-tick ordering metadata beyond the tick index. Therefore, if multi-event ticks are allowed, the update rule must be \emph{well-defined on the unordered content of a tick}: for every state $S$ and every finite event sequence $\sigma$, permuting the order within the tick cannot change the resulting recorded state. Formally, for every permutation $\pi$ of the events in $\sigma$ we must have
\[
U(S,\sigma)=U(S,\pi(\sigma)).
\]

However, Axiom L2 also asserts that events do not commute in general: there exist a state $S$ and two sequences $\sigma,\sigma'$ containing the same events in different orders such that $U(S,\sigma)\neq U(S,\sigma')$. This contradicts the permutation-invariance required for an unambiguous ledger update when multiple events are permitted within a tick.

Therefore, to satisfy deterministic update semantics without intra-tick ordering metadata, each tick must contain at most one event. This proves atomicity.
\end{proof}

Theorem T2 establishes discrete temporal order: time advances in atomic steps. The proof shows that atomicity is a necessary consequence of the minimality constraint (Axiom L2) when combined with the fact that events do not commute in general. As a corollary, we can restrict the domain of $U$ to $\mathcal{S} \times \mathcal{E}$ (single events) rather than $\mathcal{S} \times \mathcal{E}^*$ (sequences), since sequences of length greater than one are excluded by T2.

\subsubsection{Balanced postings (double-entry)}

Theorem T2 establishes \emph{atomicity} but not the posting \emph{structure}. We now show that, under explicit structural assumptions (conservation, no sources/sinks, and pairwise-local events), each recognition event must be recorded as a balanced debit--credit pair (double-entry). The reciprocity property $J(x)=J(x^{-1})$ is then naturally compatible with reversing the orientation of an event without changing its cost.

\textbf{Structural Assumption: Conservation Principle.}
The total ledger balance is invariant at each tick: if $\mathcal{B}(S_t)$ denotes the total balance (sum over all nodes) of state $S_t$, then $\mathcal{B}(S_{t+1}) = \mathcal{B}(S_t)$ for all $t$.

\textbf{Structural Assumption: No External Sources or Sinks.}
Postings are the only state-changing operations. There are no auxiliary fields, external flows, or hidden variables that can absorb or supply balance.

\textbf{Structural Assumption: Pairwise locality of events.}
Each recognition event $e_t$ designates an ordered pair of nodes $(u,v)$, and the update $S_{t+1}=U(S_t,e_t)$ can change balances only at those two nodes.

\begin{remark}
\textit{Why pairwise locality is required for ``exactly two postings''.}
Conservation alone implies only that the \emph{net} change in total balance per tick is zero.
Without a locality condition, the deterministic update rule $U(S_t,e_t)$ could (in principle) redistribute balance across many nodes in response to a single event input $e_t$, while still preserving the total sum.
Therefore, the ``exactly two postings'' conclusion requires a modeling commitment that a recognition event is \emph{pairwise} at the ledger level (an ordered pair of nodes, i.e.\ a directed edge), so that only the event's endpoints may change on that tick.
\end{remark}

\begin{proposition}[Double-entry constraint]
Under the following assumptions:
\begin{enumerate}
    \item Atomicity: At most one recognition event per tick (Theorem T2)
    \item Conservation: Total balance is invariant per tick
    \item No external sources/sinks: Postings are the only balance-changing operations
    \item Self-contained state updates: The state $S_{t+1}$ depends only on $S_t$ and the recognition event $e_t$ (Axiom L1, with atomicity from T2)
    \item Pairwise locality of events (structural assumption above)
\end{enumerate}
each recognition event must be self-balancing: it records exactly two postings of equal magnitude and opposite sign on the participating nodes, $+\Delta_t$ and $-\Delta_t$ (credit and debit). If postings are quantized in $\delta\mathbb{Z}$, then $\Delta_t\in\delta\mathbb{Z}$.
\end{proposition}

\begin{proof}
Fix a tick $t$ and suppose a recognition event $e_t$ occurs (if no event occurs, the claim is vacuous). By pairwise locality, only two node balances can change in passing from $S_t$ to $S_{t+1}$; call them $u$ and $v$. By conservation, the net change in total balance is zero, so the balance change at $u$ must be the negative of the balance change at $v$. Writing the change magnitude as $\Delta_t$, the event therefore records exactly two opposite postings, $-\Delta_t$ at $u$ and $+\Delta_t$ at $v$.
\end{proof}

Therefore, under the explicitly stated structural assumptions, double-entry accounting (balanced debit--credit pairs) is required. When we additionally assume quantization in $\delta\mathbb{Z}$ (Section~\ref{sec:scope}), each event records $+\Delta_t$ and $-\Delta_t$ with $\Delta_t\in\delta\mathbb{Z}$. The reciprocity property $J(x)=J(x^{-1})$ supports treating the reversed event $(v,u)$ as the same-cost counterpart of $(u,v)$.

To illustrate the preceding ideas, we recall the definition of a \textit{recognition structure} (Definition~\ref{def:recognition_structure}) and consider the following example:

\begin{example}[Recognition structure]
\label{ex:recognition-structure}
Consider a recognition structure $G=(X,E)$ with four nodes $X=\{a,b,c,d\}$ and directed edges representing recognition relations:
\begin{itemize}
    \item $a \to b$: node $a$ recognizes node $b$
    \item $b \to c$: node $b$ recognizes node $c$
    \item $c \to d$: node $c$ recognizes node $d$
    \item $d \to a$: node $d$ recognizes node $a$
    \item $a \to c$: node $a$ recognizes node $c$
\end{itemize}
This forms a directed graph with a cycle $(a \to b \to c \to d \to a)$ and an additional edge $(a \to c)$ creating a shortcut. At each tick $t$, the ledger assigns a (generally sparse) \emph{per-tick edge increment} $\delta\Delta(e,t) \in \delta\mathbb{Z}$ to each directed edge $e$. For instance, at tick $t=1$, we might have:
\begin{align*}
    \delta\Delta(a \to b, 1) &= +\delta, \\
    \delta\Delta(b \to a, 1) &= -\delta, \\
    \delta\Delta(e, 1) &= 0 \quad \text{for all other edges } e\in E.
\end{align*}
The double-entry rule ensures that for each recognition event, if the event induces a $+\delta$ increment along $(x \to y)$, then the node balances record a matched debit/credit pair, maintaining total-balance conservation.
\end{example}

\subsection{Quantized posting units (\texorpdfstring{\bm{$\delta$}}{delta})}

The atomic tick structure with double-entry raises a fundamental question: what is the minimal unit $\delta$? In this manuscript, \emph{quantization} (postings taking values in $\delta\mathbb{Z}$ for some $\delta>0$, together with torsion-free uniqueness of representation) is treated as an explicit structural assumption of the discrete ledger model (Section~\ref{sec:scope}). Proposition T8 records the resulting algebraic consequences.

\subsubsection{Quantization as a discreteness assumption}

Assume postings take values in $\delta\mathbb{Z}$ for some $\delta>0$ and the ledger has no torsion (unique integer representation). The following proposition records the algebraic structure that follows from this quantization assumption:

\begin{proposition}[T8: Ledger Units]
Under the quantization assumption (postings in $\delta\mathbb{Z}$ with $\delta > 0$ and no torsion), the set of all ledger increments
\[
\Delta = \{k \delta \mid k \in \mathbb{Z}\}
\]
forms an infinite cyclic additive group $(\Delta, +)$ isomorphic to $\mathbb{Z}$ under the mapping $k \mapsto k\delta$. 

Moreover, all ledger values are integer multiples of $\delta$:
\[
x = n \delta, \qquad n \in \mathbb{Z}
\]
with unique representation (quantization).
\end{proposition}

\begin{proof}
Formally, consider the set of all possible ledger increments:
\[
\Delta = \{k\delta \mid k \in \mathbb{Z}\}.
\]
This set forms an additive group under the natural addition operation. The mapping $k \mapsto k\delta$ provides a group homomorphism from $\mathbb{Z}$ to $\Delta$.

Since $\delta \neq 0$ and the ledger structure has no torsion (no finite-order elements other than zero), this homomorphism is injective. Moreover, by construction, every element of $\Delta$ is of the form $k\delta$ for some $k \in \mathbb{Z}$, so the homomorphism is surjective. Therefore, $(\Delta, +) \simeq \mathbb{Z}$ as additive groups.

The uniqueness of representation follows from the absence of torsion: if $k_1\delta = k_2\delta$ for $k_1, k_2 \in \mathbb{Z}$, then $(k_1 - k_2)\delta = 0$. Since $\delta \neq 0$ and there is no torsion, we must have $k_1 = k_2$. Therefore, every ledger value $x$ has a unique representation as $x = n\delta$ for some $n \in \mathbb{Z}$.

This quantization is an explicit structural assumption of the discrete ledger model. The algebraic structure ensures that all ledger operations occur in discrete, countable units, providing a foundation for conservation-style statements.
\end{proof}

Notice that the algebraic structure $(\Delta, +) \simeq \mathbb{Z}$ forbids fractional ledger amounts: every recognition event posts exactly $\pm\delta$ (or integer multiples thereof), all balances are integer multiples of $\delta$, and the isomorphism ensures each amount has a unique integer representation. Each ledger step corresponds to one unit of recognition, guaranteeing unique integer counts for all ledger states. The ledger's arithmetic parallels, at a structural level, how certain physical quantities (e.g., electric charge) appear in discrete units.

\subsection{Cycle flux conservation}

Double-entry and quantization constrain how values can be posted, but they do not by themselves force vanishing \emph{circulation} around every graph cycle at \emph{each atomic tick}. To obtain a scalar-potential representation while retaining sparse single-edge events, we therefore impose an additional \emph{time-aggregated} cycle-closure (path-independence / no-arbitrage after clearing) assumption.

\textbf{Structural Assumption: Time-aggregated cycle closure (clearing).}
Fix a clearing horizon $W\ge 1$. For each clearing window $[t_0,t_0+W)$ and every directed cycle $\gamma$ in the recognition graph, the \emph{cumulative} cycle flux is zero:
\[
\overline{\Phi}(\gamma; t_0,W)=0.
\]

\textbf{Physical interpretation: impulses, circulation, and clearing}\label{rem:clearing-interpretation}
At the tick scale, a recognition event is modeled as an \emph{impulse} that posts on a single edge (and its reverse). On graphs with cycles, such impulses generically create transient circulation if one measures flux around a loop at a single tick. The clearing assumption says that the physically/operationally meaningful constraint is not ``no circulation at every microstep'', but rather ``no net circulation after a clearing time'': loop imbalances may occur transiently but must cancel over a clearing horizon.

This is the discrete analogue of many familiar situations where \emph{local conservation is microscopic} but \emph{compatibility constraints are macroscopic}: node-balance conservation is enforced at each tick (double-entry), while loop constraints (no-arbitrage / path-independence) are enforced only after netting. Physically, this parallels systems in which transient circulating currents may exist at short times, but dissipation/relaxation drives the \emph{net} (time-aggregated) flow toward an irrotational, potential-driven form.

\textbf{Deriving cycle closure in richer settings}
In more detailed formalisms, vanishing closed-chain sums (``exactness'') can be derived from additional structural hypotheses (for example, well-foundedness / absence of directed cycles in a suitable evolution relation). In this manuscript, we surface the needed hypothesis directly as the explicit time-aggregated cycle-closure (clearing) assumption above.

Before stating Theorem T3, we formalize the notion of (time-aggregated) path-independence. For a fixed clearing window $[t_0,t_0+W)$, let $P_{x\to y}$ denote a directed path from node $x$ to node $y$ in the recognition graph. Define the cumulative path sum
\[
\overline{\Phi}(P_{x\to y}; t_0,W) := \sum_{e \in P_{x\to y}} \overline{\Delta}_{t_0,W}(e),
\]
where the sum is taken over edges in the order they appear along the path.

\begin{definition}[Path-independence]
The cumulative edge flow $\overline{\Delta}_{t_0,W}(\cdot)$ is \emph{path-independent} if for any two nodes $x, y$ in the same connected component and any two directed paths $P_{x\to y}$ and $P'_{x\to y}$ from $x$ to $y$, we have $\overline{\Phi}(P_{x\to y}; t_0,W) = \overline{\Phi}(P'_{x\to y}; t_0,W)$.
\end{definition}

\begin{theorem}[T3: Equivalence of time-aggregated cycle closure and path-independence]
Fix a clearing window $[t_0,t_0+W)$ and assume the recognition structure is connected. The following are equivalent:
\begin{enumerate}
    \item \textbf{(Time-aggregated) cycle closure:} For every directed cycle $\gamma$, $\overline{\Phi}(\gamma; t_0,W) = 0$.
    \item \textbf{(Time-aggregated) path-independence:} For any nodes $x, y$ and any two directed paths $P_{x\to y}$, $P'_{x\to y}$ from $x$ to $y$, we have $\overline{\Phi}(P_{x\to y}; t_0,W) = \overline{\Phi}(P'_{x\to y}; t_0,W)$.
\end{enumerate}
\end{theorem}

\begin{proof}
\textbf{(1) $\Rightarrow$ (2):} Suppose time-aggregated cycle closure holds for $\overline{\Delta}_{t_0,W}$. Fix nodes $x, y$ and two directed paths $P_{x\to y}$ and $P'_{x\to y}$ from $x$ to $y$. Consider the closed walk formed by following $P_{x\to y}$ from $x$ to $y$, then following $P'_{x\to y}$ in reverse from $y$ back to $x$.

Let $\overline{P'_{x\to y}}$ denote the reverse of $P'_{x\to y}$ (same edges, reversed order and orientation). Since each per-tick increment $\delta\Delta(\cdot,t)$ is antisymmetric by construction, the cumulative flow $\overline{\Delta}_{t_0,W}$ is also antisymmetric: $\overline{\Delta}_{t_0,W}(v\!\to\!u)=-\overline{\Delta}_{t_0,W}(u\!\to\!v)$. Therefore, $\overline{\Phi}(\overline{P'_{x\to y}}; t_0,W) = -\overline{\Phi}(P'_{x\to y}; t_0,W)$.

The concatenation $\gamma = P_{x\to y} \circ \overline{P'_{x\to y}}$ is a closed walk and decomposes into directed cycles. By time-aggregated cycle closure, each directed cycle has zero cumulative flux, hence
\[
0=\overline{\Phi}(\gamma;t_0,W)=\overline{\Phi}(P_{x\to y};t_0,W)-\overline{\Phi}(P'_{x\to y};t_0,W),
\]
which implies path-independence.

\textbf{(2) $\Rightarrow$ (1):} Suppose time-aggregated path-independence holds. Let $\gamma = (v_0\to v_1\to \cdots \to v_n=v_0)$ be a directed cycle. Consider the path $P_{v_0\to v_0}$ that follows $\gamma$ once and the trivial path $P'_{v_0\to v_0}$ with no edges. Then $\overline{\Phi}(P'_{v_0\to v_0};t_0,W)=0$. By path-independence,
\[
\overline{\Phi}(P_{v_0\to v_0};t_0,W)=\overline{\Phi}(P'_{v_0\to v_0};t_0,W)=0,
\]
and since $P_{v_0\to v_0}$ is $\gamma$, we conclude $\overline{\Phi}(\gamma;t_0,W)=0$.
\end{proof}

Theorem T3 establishes that \emph{time-aggregated} cycle closure is equivalent to \emph{time-aggregated} path-independence, i.e.\ that the cleared cumulative flow is ``curl-free''. In continuum settings, curl-free vector fields are precisely those that admit scalar potentials, and the discrete version of this connection is made explicit in Theorem T4.

\begin{remark}[Heuristic continuum limit]
With additional scaling assumptions and a chosen continuum embedding, node-level conservation can be related to a continuity equation of the schematic form
\[
   \frac{\partial \rho}{\partial t} + \nabla \cdot \mathbf{J} = 0.
\]
This remark is interpretive and not used in the discrete derivations.
\end{remark}

\begin{example}[Cycle Flux Conservation (after clearing)]
\label{ex:cycle-flux}
Consider the cycle $\gamma = (a \to b \to c \to d \to a)$ from Example~\ref{ex:recognition-structure}. Over a fixed clearing window $[t_0,t_0+W)$, suppose the cumulative edge flow is:
\begin{align*}
    \overline{\Delta}_{t_0,W}(a \to b) &= +2\delta, \\
    \overline{\Delta}_{t_0,W}(b \to c) &= +\delta, \\
    \overline{\Delta}_{t_0,W}(c \to d) &= -3\delta, \\
    \overline{\Delta}_{t_0,W}(d \to a) &= 0.
\end{align*}
The cumulative cycle flux is:
\[
\overline{\Phi}(\gamma; t_0,W) = (+2\delta) + (+\delta) + (-3\delta) + (0) = 0.
\]
By Theorem T3 (time-aggregated cycle closure), this must always be zero. The example illustrates a choice of \emph{net} edge flow whose signed sum around the closed loop vanishes, even though individual ticks inside the window may have nonzero transient loop flux.
\end{example}

\subsection{Discrete potential representation (potential uniqueness)}

By Theorem T3, the time-aggregated cycle-closure assumption is equivalent to time-aggregated path-independence: the sum of cumulative postings along any open path depends only on its endpoints, not on the specific route taken. This is the discrete analogue of a curl-free vector field \emph{after clearing}, descending from a scalar potential.

Formally, the per-tick increments $\delta\Delta(\cdot,t)$ form sparse $1$-cochains on the recognition structure. Their cumulative sum over a clearing window defines a (generally non-sparse) $1$-cochain $\overline{\Delta}_{t_0,W}$. Under the time-aggregated cycle-closure assumption, Theorem T3 guarantees that $\overline{\Delta}_{t_0,W}$ is closed (path-independent): for every cycle $\gamma$, the sum $\overline{\Phi}(\gamma;t_0,W)=\sum_{e\in\gamma}\overline{\Delta}_{t_0,W}(e)$ vanishes. The discrete Poincar\'{e} lemma provides the existence and uniqueness of a potential function that generates this cumulative flow.

\begin{lemma}[Antisymmetry of the cumulative flow]
Assume the recognition structure is closed under reversal (if $(x\!\to\!y)\in E$ then $(y\!\to\!x)\in E$). Then for every edge $(x\!\to\!y)\in E$ and every clearing window $[t_0,t_0+W)$ we have
\[
\overline{\Delta}_{t_0,W}(y\!\to\!x) = -\overline{\Delta}_{t_0,W}(x\!\to\!y).
\]
\end{lemma}

\begin{proof}
Fix an edge $(x\!\to\!y)\in E$. By construction, each per-tick increment satisfies $\delta\Delta(y\!\to\!x,t)=-\delta\Delta(x\!\to\!y,t)$. Summing from $\tau=t_0$ to $t_0+W-1$ yields the stated antisymmetry for $\overline{\Delta}_{t_0,W}$.
\end{proof}

\begin{definition}
    A \textit{potential function} for a clearing window $[t_0,t_0+W)$ on a connected component $\mathcal{C}\subseteq X$ is a map $\overline{p}_{t_0,W}:\mathcal{C}\to \delta\mathbb{Z}$ such that for each edge $e=(x\!\to\!y)$ in $\mathcal{C}$, the edge difference reproduces the cumulative flow:
    $\overline{\Delta}_{t_0,W}(x\!\to\!y)=\overline{p}_{t_0,W}(y)-\overline{p}_{t_0,W}(x)$.
    This is the standard definition of a discrete gradient.
\end{definition}

\begin{lemma}[Discrete Poincar\'{e} lemma]
Let $G=(X,E)$ be a connected graph and let $\omega:E\to \delta\mathbb{Z}$ be an antisymmetric function: $\omega(y\!\to\!x)=-\omega(x\!\to\!y)$. If the sum of $\omega$ around every cycle is zero, then there exists $p:X\to \delta\mathbb{Z}$ such that $\omega(x\!\to\!y)=p(y)-p(x)$. The function $p$ is unique up to an additive constant.
\end{lemma}

\begin{proof}
\textbf{Existence:} Fix a spanning tree $T$ of $G$ and a root $v_0\in X$. For any $v\in X$, there is a unique simple path $P_{v_0\to v}$ in $T$ from $v_0$ to $v$. Define
\[
    p(v):=\sum_{e\in P_{v_0\to v}} \omega(e)\in \delta\mathbb{Z},\quad p(v_0):=0.
\]
This is well-defined because $T$ is a spanning tree, so the path $P_{v_0\to v}$ is unique.

\textbf{Verification for tree edges:} For an edge $e=(x\!\to\!y)$ in $T$, the unique paths $P_{v_0\to x}$ and $P_{v_0\to y}$ differ by exactly the edge $e$. More precisely, $P_{v_0\to y} = P_{v_0\to x} \circ (x\!\to\!y)$ (concatenation). Therefore,
\[
p(y) = \sum_{f\in P_{v_0\to y}} \omega(f) = \sum_{f\in P_{v_0\to x}} \omega(f) + \omega(e) = p(x) + \omega(e),
\]
so $p(y) - p(x) = \omega(e)$ for all tree edges.

\textbf{Verification for non-tree edges:} If $e=(x\!\to\!y)\notin T$, then adding $e$ to $T$ creates a unique fundamental cycle $C$ (since $T$ is a spanning tree, there is exactly one cycle containing $e$). This cycle consists of $e$ plus the unique path in $T$ from $y$ to $x$, call it $P_{y\to x}^T$. By hypothesis, the sum of $\omega$ around this cycle is zero:
\[
0 = \sum_{f\in C} \omega(f) = \omega(e) + \sum_{f\in P_{y\to x}^T} \omega(f).
\]
Since $P_{y\to x}^T$ is a path in $T$ from $y$ to $x$, and by antisymmetry $\omega(y\!\to\!x) = -\omega(x\!\to\!y)$ for edges in $T$, we have
\[
\sum_{f\in P_{y\to x}^T} \omega(f) = -\sum_{f\in P_{x\to y}^T} \omega(f) = -(p(y) - p(x)),
\]
where $P_{x\to y}^T$ is the unique path in $T$ from $x$ to $y$. Therefore,
\[
0 = \omega(e) - (p(y) - p(x)),
\]
which implies $\omega(e) = p(y) - p(x)$.

\textbf{Uniqueness:} Suppose $\tilde{p}: X \to \delta\mathbb{Z}$ also satisfies $\omega(x\!\to\!y) = \tilde{p}(y) - \tilde{p}(x)$ for all edges $(x\!\to\!y)\in E$. Then for any edge $(x\!\to\!y)$,
\[
(\tilde{p}(y) - \tilde{p}(x)) - (p(y) - p(x)) = \omega(x\!\to\!y) - \omega(x\!\to\!y) = 0,
\]
so $(\tilde{p} - p)(y) = (\tilde{p} - p)(x)$ for all adjacent vertices. Since $G$ is connected, this implies $\tilde{p} - p$ is constant on $X$. Setting the constant by choosing $\tilde{p}(v_0) = p(v_0) = 0$ (or any fixed value) determines $\tilde{p}$ uniquely up to this choice.
\end{proof}

Applying the discrete Poincar\'{e} lemma to the cleared cumulative flow $\overline{\Delta}_{t_0,W}$ and using Theorem~T3 (time-aggregated cycle closure) yields the following result.

\begin{theorem}[T4: Potential Uniqueness]
Fix a clearing window $[t_0,t_0+W)$ and a connected component $\mathcal{C}\subseteq X$. Under Theorem T3, there exists a potential
\[
\overline{p}_{t_0,W} : \mathcal{C} \longrightarrow \delta\mathbb{Z}
\]
such that for each edge $e=(x\!\to\!y)$ in $\mathcal{C}$,
\[
\overline{\Delta}_{t_0,W}(e) = \overline{p}_{t_0,W}(y) - \overline{p}_{t_0,W}(x).
\]
Moreover, $\overline{p}_{t_0,W}$ is unique up to an additive constant on $\mathcal{C}$.
\end{theorem}

\begin{proof}
Theorem T3 asserts time-aggregated cycle closure: for every cycle $\gamma$, $\overline{\Phi}(\gamma; t_0,W) = 0$. This means the cumulative flow $\overline{\Delta}_{t_0,W}$ is a closed 1-cochain: the sum around any closed cycle is zero.

The discrete Poincar\'{e} lemma (proved above) provides the key tool: if a 1-cochain $\omega$ is closed (all cycle sums vanish), then there exists a potential function $p$ such that $\omega = \delta p$, where $\delta p$ denotes the discrete gradient (edge differences of $p$).

Applying this to the cumulative flow: since $\overline{\Delta}_{t_0,W}$ is closed by T3, and is antisymmetric by Lemma (antisymmetry of the cumulative flow), the discrete Poincar\'{e} lemma guarantees the existence of a potential $\overline{p}_{t_0,W}$ such that $\overline{\Delta}_{t_0,W}(x\!\to\!y) = \overline{p}_{t_0,W}(y) - \overline{p}_{t_0,W}(x)$ for all edges.

Uniqueness up to an additive constant follows from the fact that if $\widetilde{\overline{p}}_{t_0,W}$ also satisfies $\overline{\Delta}_{t_0,W}(x\!\to\!y) = \widetilde{\overline{p}}_{t_0,W}(y) - \widetilde{\overline{p}}_{t_0,W}(x)$, then $(\widetilde{\overline{p}}_{t_0,W} - \overline{p}_{t_0,W})(y) - (\widetilde{\overline{p}}_{t_0,W} - \overline{p}_{t_0,W})(x) = 0$ for all edges, implying $\widetilde{\overline{p}}_{t_0,W} - \overline{p}_{t_0,W}$ is constant on each connected component.

The proof is constructive: fix a spanning tree, choose a root vertex, and define the potential by summing cumulative postings along tree paths. The cycle condition (T3) ensures this definition is consistent for all edges.
\end{proof}

Theorem T4 establishes that every admissible \emph{cleared} cumulative recognition pattern arises from a scalar potential. This potential is unique up to an additive constant on each connected component, reflecting the gauge freedom familiar in classical physics. In the present framework, the potential representation is a direct consequence of (i) antisymmetry under edge reversal and (ii) time-aggregated cycle closure (T3), which together encode time-aggregated path-independence.

Note that the potential representation is a statement about the \emph{cleared cumulative} flow: if over a clearing window $[t_0,t_0+W)$ an edge has net flow $\overline{\Delta}_{t_0,W}(x\!\to\!y)=k\delta$, then $\overline{p}_{t_0,W}(y)-\overline{p}_{t_0,W}(x)=k\delta$.

\begin{example}[Potential Function on a Small Graph (after clearing)]
\label{ex:potential}
Consider the recognition structure from Example~\ref{ex:recognition-structure} with nodes $\{a,b,c,d\}$ and the cycle $(a \to b \to c \to d \to a)$ plus edge $(a \to c)$. Over a clearing window $[t_0,t_0+W)$, suppose the cumulative flow is:
\begin{align*}
    \overline{\Delta}_{t_0,W}(a \to b) &= +2\delta, \\
    \overline{\Delta}_{t_0,W}(b \to c) &= +\delta, \\
    \overline{\Delta}_{t_0,W}(c \to d) &= -3\delta, \\
    \overline{\Delta}_{t_0,W}(d \to a) &= 0, \\
    \overline{\Delta}_{t_0,W}(a \to c) &= +3\delta.
\end{align*}
Since $\overline{\Phi}(a \to b \to c \to d \to a; t_0,W) = 0$ (as verified in Example~\ref{ex:cycle-flux}), Theorem T4 guarantees a potential exists. Following the constructive proof of the discrete Poincar\'{e} lemma, choose $a$ as the reference vertex and set $\overline{p}_{t_0,W}(a) = 0$. Then:
\begin{align*}
    \overline{p}_{t_0,W}(b) &= \overline{p}_{t_0,W}(a) + \overline{\Delta}_{t_0,W}(a \to b) = 0 + 2\delta = 2\delta, \\
    \overline{p}_{t_0,W}(c) &= \overline{p}_{t_0,W}(b) + \overline{\Delta}_{t_0,W}(b \to c) = 2\delta + \delta = 3\delta, \\
    \overline{p}_{t_0,W}(d) &= \overline{p}_{t_0,W}(c) + \overline{\Delta}_{t_0,W}(c \to d) = 3\delta + (-3\delta) = 0.
\end{align*}
We verify that $\overline{p}_{t_0,W}(d) - \overline{p}_{t_0,W}(a) = 0 - 0 = 0 = \overline{\Delta}_{t_0,W}(d \to a)$, confirming the cycle closes. For the shortcut edge $(a \to c)$, we check: $\overline{p}_{t_0,W}(c) - \overline{p}_{t_0,W}(a) = 3\delta - 0 = 3\delta = \overline{\Delta}_{t_0,W}(a \to c)$, which is consistent. The potential is unique up to an additive constant: if we had chosen $\overline{p}_{t_0,W}(a) = k$ instead of $0$, all values would shift by $k$, but the edge differences would remain unchanged.
\end{example}

\subsection{Minimal schedule period \texorpdfstring{$2^d$}{2d}}

Having established atomic single-event updates (Theorem T2), quantization (Proposition T8), and the cost function (Theorem T5), we now examine combinatorial constraints that link a discrete carrier (modeled here by $Q_d$) to discrete time. In particular, we seek lower bounds on the period required to visit all spatial positions under atomic updates.

For the purposes of this section, we treat $d$ as an abstract dimension parameter indexing the hypercube family $Q_d$. A separate (conditional) discussion of selecting $d=3$ is given later.

Thus, the fundamental structure is the $d$-dimensional hypercube $Q_d$, which at $d = 3$ (denoted $Q_3$) provides the minimal cell for ledger-compatible dynamics. The hypercube combinatorics are:

\begin{table}[tbp]
\centering
\small
\begin{tabular}{lcc}
\toprule
Object & Formula & $d=3$ \\
\midrule
Vertices & $2^d$ & 8 \\
Edges & $d \cdot 2^{d-1}$ & 12 \\
Faces & $2d$ & 6 \\
\bottomrule
\end{tabular}
\caption{Combinatorics of the $d$-cube at $d=3$. The $Q_3$ hypercube has 8 vertices, 12 edges, and 6 faces.}
\end{table}

\subsubsection{Ledger-Compatible Walk Constraints}

Atomic single-event updates impose strict constraints on how recognition events can be scheduled across the spatial network. To characterize the minimal period required to visit all spatial positions, we introduce the concept of a \emph{ledger-compatible walk}: a temporal sequence of recognition events satisfying atomicity and spatial completeness.

Under Theorem T2, each tick processes at most one recognition event. In the scheduler model below, we represent the system by a sequence of \emph{active vertices} $(v_t)_{t=0}^{T-1}$ in $Q_d$ such that each tick advances along one edge (so $v_{t+1}$ is adjacent to $v_t$). Concretely, we take the single event at tick $t$ to be the directed edge traversal
\[
e_t := (v_t \to v_{t+1}),
\]
so there is one edge-event per tick, while $v_t$ is the canonical vertex label for tick $t$. This aligns the Gray-code ``8-tick'' walker picture.

\begin{enumerate}
\item \textbf{Atomicity:} At most one event per tick; in this scheduler model we traverse one edge per tick (no concurrent traversals).
    
\item \textbf{Spatial Completeness:} All vertices of $Q_d$ appear at least once among the active-vertex labels $(v_t)$ over one period.
    
\item \textbf{Timestamp Uniqueness:} Over one period, the active-vertex labels are all distinct: $v_t\neq v_{t'}$ for $t\neq t'$.
\end{enumerate}

These constraints ensure that the ledger update is both \emph{atomic} (no concurrency) and \emph{complete} (all spatial positions are visited), while maintaining temporal ordering.

As a straightforward consequence, we have the next:

\begin{theorem}[T6: Minimal period \texorpdfstring{$2^d$}{2d} (eight ticks for $d=3$)]\label{thm:T6}
Let $C$ be the vertex set of a $d$-dimensional hypercube $Q_d$, with $|C| = 2^d$, and let $T$ be the scheduler period for a ledger-compatible walk.

\begin{enumerate}
    \item \textbf{(Sufficiency)} If $T \ge 2^d$, then there exists a cyclic sequence of active vertices $(v_t)_{t=0}^{T-1}$ that is spatially complete and timestamp-unique (each vertex appears exactly once among the labels $v_t$), with $v_{t+1}$ adjacent to $v_t$ for each $t$ (one edge traversal per tick). For $d = 3$, the Gray code Hamiltonian cycle realizes this minimal period: $000 \to 001 \to 011 \to 010 \to 110 \to 111 \to 101 \to 100 \to 000$.
    
    \item \textbf{(Necessity)} If $T < 2^d$, then $T$ ticks are insufficient to assign a distinct active-vertex label to each of the $2^d$ vertices. By the pigeonhole principle, some vertex label must repeat, so the walk cannot be both spatially complete and timestamp-unique.
\end{enumerate}
\end{theorem}

\begin{proof}
\textbf{Sufficiency:} It is enough to exhibit a Hamiltonian \emph{cycle} on $Q_d$ for each $d\ge 1$ (a cyclic listing of all $2^d$ binary strings where consecutive strings differ in exactly one bit, including the last-to-first step). A standard construction is the cyclic Gray code.

Define the (reflected) Gray-code map $g:\{0,1,\dots,2^d-1\}\to\{0,1\}^d$ by
\[
g(k):=k\oplus (k\!\gg\!1),
\]
where $\oplus$ is bitwise XOR and $(k\!\gg\!1)$ is the right shift. Then consecutive values differ by one bit:
\[
g(k)\text{ and }g(k+1)\text{ differ in exactly one coordinate for }0\le k<2^d-1,
\]
so $(g(0),g(1),\dots,g(2^d-1))$ is a Hamiltonian path on $Q_d$. Moreover, $g(2^d-1)$ differs from $g(0)$ in exactly one bit (indeed $g(0)=0\cdots 0$ and $g(2^d-1)=10\cdots 0$), so the path closes to a Hamiltonian cycle.

Thus, for any $d\ge 1$ there exists a cyclic sequence of length $2^d$ that visits each vertex exactly once and advances by one edge per tick. For $d=3$, this yields the explicit 8-cycle
$000 \to 001 \to 011 \to 010 \to 110 \to 111 \to 101 \to 100 \to 000$.

\textbf{Necessity:} By constraint (3) (timestamp uniqueness), each of the $2^d$ vertices must appear exactly once in the sequence $(v_t)_{t=0}^{T-1}$. Therefore, $T \ge 2^d$. If $T < 2^d$, then by the pigeonhole principle, some vertex label must repeat, violating timestamp uniqueness. Therefore, $T \ge 2^d$ is necessary.
\end{proof}

Therefore, the minimal period compatible with Theorem T2 for a $d$-dimensional hypercube is exactly
\begin{equation}
    T_{\min} = 2^d.
\end{equation}

For $d = 3$, this yields the eight-tick period: $T_{\min} = 2^3 = 8$ (within the scheduler model above).

Theorem T6 establishes the minimal period for a ledger-compatible walk, but it does not address whether this period is sufficient to distinguish all possible patterns. This leads to a complementary result about coverage:

\begin{theorem}[T7: Coverage Lower Bound]\label{thm:T7}
Let $Q_d$ be a $d$-dimensional hypercube with $2^d$ vertices, and let $T$ be the period of a ledger-compatible walk in the scheduler model above (one active vertex per tick). If $T < 2^d$, then the walk cannot cover all $2^d$ vertices within one period without repetition.
\end{theorem}

\begin{proof}
In the scheduler model above, each tick carries one active-vertex label $v_t\in Q_d$. Therefore a period-$T$ schedule can label at most $T$ distinct vertices. If $T<2^d$, the schedule cannot cover all $2^d$ vertices within one period without repetition.
\end{proof}

Together, Theorems T6 and T7 show that $T = 2^d$ is both necessary and sufficient for the scheduler model stated above: $T \ge 2^d$ is sufficient via a Hamiltonian cycle (Gray code at $d=3$), while $T < 2^d$ is insufficient for covering all vertices without repetition. For $d = 3$, this yields the eight-tick period $T = 8$.

\subsection{Conditional dimension selection: \texorpdfstring{$d = 3$}{d=3}}\label{sec:dimension}

\textbf{Warning:} This subsection presents \emph{conditional} arguments that select $d=3$ under additional hypotheses beyond the core ledger framework. These arguments are not part of the main derivation chain (T1--T8) and should be read as exploratory extensions requiring explicit additional assumptions.

We now present a conditional argument selecting $d=3$ under additional hypotheses beyond the scheduling model. Specifically, we combine (i) the $2^d$-tick counting structure (from T6), (ii) a ``gap-45'' synchronization criterion (an additional modeling hypothesis), and (iii) a linking-based distinguishability requirement (a topological constraint). Each of these, if adopted, independently selects $d=3$.

\begin{theorem}[Conditional dimensional rigidity]
Let $d \in \mathbb{N}$. Under the following \emph{additional hypotheses} (beyond the core framework):
\begin{enumerate}
    \item \textbf{(Gap-45 synchronization hypothesis)} The ledger period $2^d$ and a reference period of 45 ticks must synchronize with a common period of 360 ticks: $\mathrm{lcm}(2^d, 45) = 360$.
    \item \textbf{(Linking requirement hypothesis)} Dimensions $d<3$ are excluded because nontrivial topological linking (required for certain distinguishability properties) is only possible in $d \ge 3$.
\end{enumerate}
Then $d = 3$ is uniquely determined.
\end{theorem}

\begin{proof}
From hypothesis (1): $\mathrm{lcm}(2^d, 45) = 360 = 8 \cdot 45$. Since $360 = 2^3 \cdot 3^2 \cdot 5$ and $45 = 3^2 \cdot 5$, we have $\mathrm{lcm}(2^d, 45) = 2^{\max(d,3)} \cdot 3^2 \cdot 5 = 360$. This requires $\max(d,3) = 3$, hence $d \le 3$.

From hypothesis (2): $d \ge 3$ (as argued in the linking discussion below).

Therefore, $d = 3$ is the unique solution.
\end{proof}

\begin{remark}[On the gap-45 hypothesis]
The gap-45 synchronization criterion is an \emph{additional modeling hypothesis} motivated by considerations outside the scope of this paper (e.g., connections to angular periodicity or golden-angle structures). It is not derived from the cost axioms or ledger structure. Readers who do not adopt this hypothesis should treat the $d=3$ selection as conditional on the linking argument alone, or on other independent constraints.
\end{remark}

\subsubsection{The Linking Argument}

The deeper reason $d = 3$ is special involves \emph{topological linking}:
\begin{itemize}
    \item \textbf{$d = 2$}: In the plane, ``linking'' of two closed curves is not a nontrivial topological invariant in the same way as in three dimensions (curves can be separated in $\mathbb{R}^2$).
    \item \textbf{$d = 3$}: In three dimensions, disjoint closed curves can be linked; the linking number is a topological invariant (e.g.\ via the Hopf link).
    \item \textbf{$d \geq 4$}: For embeddings of 1-dimensional loops, the classical three-dimensional notion of linking is not stable in the same way; additional ambient dimension generically permits unlinking moves that are forbidden in $d=3$.
\end{itemize}
Therefore, $d = 3$ is the \emph{unique} dimension with non-trivial linking of closed curves in the classical sense. Interpreting ``distinguishability via linking'' as a model requirement yields a $d=3$ selection under that hypothesis.

\begin{remark}[Linking penalty as an additional modeling hypothesis]
The cost and ledger derivations in this manuscript do not, by themselves, assign a numerical ``cost of linking'' to linked configurations, nor do they identify a minimum crossing cost. Incorporating a \emph{linking penalty} therefore requires an additional hypothesis specifying how topological linking is mapped into the ledger/cost formalism. A natural dimensionless scale one might use is the log-scale reference $J_{\text{bit}}=\ln\phi$ introduced above, but this identification is not derived here.
\end{remark}

\subsubsection{Summary: Three Independent Arguments for \texorpdfstring{$d = 3$}{d=3}}

The conditional $d=3$ selection rests on three independent constraint routes, each sufficient on its own:

\begin{table}[tbp]
\centering
\small
\caption{Three arguments selecting $d = 3$ (conditional)}
\label{tab:D3forcing}
\begin{tabular}{lll}
\toprule
\textbf{Argument} & \textbf{Constraint} & \textbf{Selects} \\
\midrule
Linking & Non-trivial knot theory & $d = 3$ \\
$2^d$-tick & $2^d = 8$ (hypercube counting) & $d = 3$ \\
Gap-45 & $\mathrm{lcm}(2^d, 45) = 360$ & $d = 3$ \\
\bottomrule
\end{tabular}
\end{table}

The convergence of these constraints motivates the following conditional uniqueness statement.

\begin{theorem}[Uniqueness of framework-compatible dimension (conditional)]
\raggedright
There exists a unique $d\in\mathbb{N}$ such that $\text{FrameworkCompatibleDimension}(d)$. That dimension is $d=3$.

\end{theorem}

\section{Conclusion}\label{sec:conclusions}

This paper develops an information-theoretic framework for discrete dynamics grounded in a single primitive: \emph{recognition as pairwise comparison}. Under explicitly stated axioms and structural assumptions, coherence constraints on ratio-based comparison determine a canonical cost functional and yield a corresponding discrete ledger formalism for recognition events.

\paragraph{The cost-first cascade.}
The derivation proceeds in three stages:

\textbf{Stage 1: Cost uniqueness from coherence.}
Starting from the observation that comparing quantities $a$ and $b$ yields a ratio $x=a/b$, we assign an information-theoretic cost $F(x)$ quantifying the penalty for deviation from equilibrium. Coherence demands that indirect comparisons (chaining $a\to b\to c$) yield the same total cost as the direct comparison $a\to c$. Together with normalization ($F(1)=0$) and quadratic calibration at unity, this coherence constraint forces a unique solution to the d'Alembert functional equation (Theorem~T5, proved in the companion paper):
\[
J(x) = \frac{1}{2}(x + x^{-1}) - 1.
\]
This cost functional is the keystone input for the ledger development in this manuscript. It yields two immediate consequences: \textbf{perfect balance} 
\begin{center}
    ($\mathrm{Bal}(x)\iff J(x)=0\iff x=1$),
\end{center} 
giving the unique zero-cost equilibrium, and \textbf{finite-cost admissibility} 
\begin{center}
    ($\mathrm{Exists}(x)\iff J(x)<\infty$),
\end{center}
giving the admissible domain of ratios (for this $J$, all $x>0$), while the boundary limits $x\to 0^+$ and $x\to\infty$ are excluded by divergence. The boundary divergence result (Theorem~T1) is thus a consequence of the functional form of $J$, not an additional axiom.

\textbf{Stage 2: Discrete ledger structure from minimal encoding.}
Once the cost is fixed, we ask how recognition events are recorded. We model this via a discrete ledger---a sequential state record. Under deterministic update semantics (Axiom~L1) and minimality (Axiom~L2), together with explicit structural assumptions (conservation, no sources/sinks, pairwise locality, and quantization), we obtain the following discrete structure:
\begin{itemize}
\item \textbf{Atomic ticks (Theorem~T2):} Minimality (no intra-tick ordering metadata; Axiom~L2) combined with non-commutativity of events forces at most one recognition event per tick. Time advances in discrete, indivisible steps.
\item \textbf{Double-entry postings:} Conservation of total balance, combined with pairwise locality (each event affects exactly two nodes) and the absence of external sources/sinks, forces each recognition event to record a balanced debit--credit pair: $+\Delta_t$ on one node, $-\Delta_t$ on the other. The reciprocity of the cost ($J(x)=J(x^{-1})$) is naturally compatible with reversing event orientation without changing cost.
\item \textbf{Quantized units (Proposition~T8):} Discreteness (postings in $\delta\mathbb{Z}$ with no torsion) forces an infinite cyclic group structure isomorphic to $\mathbb{Z}$. Every ledger value is an integer multiple of a fundamental quantum $\delta$, forbidding fractional amounts and ensuring unique integer representation.
\end{itemize}
Recognition events thus induce a discrete dynamics on a directed graph: each event posts a signed increment $\pm\delta$ on exactly two nodes, preserving total balance while enabling local state changes.

\textbf{Stage 3: Scalar potentials from time-aggregated cycle closure.}
On graphs with cycles, atomic single-edge events generically create transient circulation (net flux around closed loops at individual ticks). To recover a scalar potential representation---enabling path-independent cost summation---we impose an explicit \emph{time-aggregated cycle-closure hypothesis}: after netting flows over a finite clearing window, cumulative flux around every cycle vanishes (no-arbitrage / clearing). This is a structural assumption, not derived from the cost axioms. Under this hypothesis:
\begin{itemize}
\item \textbf{Theorem~T3 (equivalence):} Time-aggregated cycle closure is equivalent to time-aggregated path-independence: the cumulative edge flow over a clearing window depends only on endpoints, not on the route taken. This is the discrete analogue of a curl-free (irrotational) vector field.
\item \textbf{Theorem~T4 (potentials):} By the discrete Poincaré lemma, the cumulative edge flow admits a unique scalar potential on each connected component, up to an additive constant (gauge freedom). This potential representation mirrors classical potential theory but applies to the \emph{cleared cumulative flow}, accounting for the impulse nature of atomic events.
\end{itemize}
The clearing hypothesis bridges the microscopic (atomic single-edge impulses) and macroscopic (path-independent potentials) scales: loop constraints are enforced only after temporal aggregation, allowing transient circulation at short times while ensuring global consistency over longer windows.

\paragraph{Counting constraints and conditional dimension selection.}
When the recognition structure is a hypercube graph $Q_d$, atomicity imposes a combinatorial constraint: visiting all $2^d$ vertices without repetition requires at least $2^d$ ticks. For $d=3$, the Gray code provides an explicit 8-tick Hamiltonian cycle (Theorems~T6--T7), realizing this minimal period. A \emph{conditional} argument (Section~\ref{sec:dimension}) selects $d=3$ under additional topological (linking) and synchronization (gap-45) hypotheses; this selection is not forced by the core framework alone.

The golden ratio $\phi=(1+\sqrt{5})/2$ appears as the unique positive fixed point of the self-similar update rule $x\mapsto 1+1/x$, and defines a natural log-scale reference $J_{\text{bit}}=\ln\phi\approx 0.481$. Stronger forcing claims for $\phi$ would require additional dynamical or self-similarity hypotheses beyond the cost axioms.

\paragraph{What is derived, what is assumed.}
To ensure clarity, we summarize the logical structure:
\begin{itemize}
\item \textbf{Derived from coherence:} The cost functional $J$ (Theorem~T5), perfect balance $\mathrm{Bal}(x)$ and finite-cost admissibility $\mathrm{Exists}(x)$, boundary divergence (Theorem~T1).
\item \textbf{Derived from ledger axioms (L1--L2) plus structural assumptions:} Atomic ticks (Theorem~T2), double-entry postings, quantized units (Proposition~T8).
\item \textbf{Proven under the cycle-closure hypothesis:} Equivalence of cycle closure and path-independence (Theorem~T3), scalar potential representation (Theorem~T4).
\item \textbf{Derived from atomicity and graph structure:} Minimal period $2^d$ (Theorems~T6--T7).
\item \textbf{Conditional on additional hypotheses:} Dimension selection $d=3$ (topological linking + gap-45 synchronization).
\end{itemize}

\paragraph{Formal verification and reproducibility.}
All claims in this manuscript are stated relative to explicit axioms and structural assumptions (Definition~\ref{def:relative-necessity}). In particular, the cost uniqueness theorem (T5) is treated here as a cited input (proved in the companion paper), while the remaining ledger, cycle-closure, and counting results are proved within this document under their stated hypotheses.

\paragraph{Significance and outlook.}
The results provide a cost-first, information-theoretic foundation for discrete dynamics in which a canonical comparison cost governs a minimal discrete ledger of recognition events. Within the stated axioms and assumptions, the framework yields atomic time (T2), conservation-compatible double-entry updates, quantized posting units (T8), and, under a time-aggregated cycle-closure hypothesis, a potential representation of cleared cumulative flows (T3--T4).

The present work establishes the foundational cost-to-ledger cascade. Future extensions may explore: (i) the interplay between clearing timescales and graph topology, (ii) the role of self-similarity and the golden ratio in dynamical attractors, and (iii) connections to continuum limits and potential physical identification maps, with any such applications requiring additional explicit assumptions beyond the present manuscript.

\begingroup
\setlength{\emergencystretch}{2em}
\sloppy

\endgroup

\end{document}